\documentclass{eptcs}
\usepackage{amsfonts}
\usepackage{amsthm}
\usepackage{amsmath}
\usepackage{amssymb}
\usepackage{blkarray}
\usepackage[table]{xcolor}
\usepackage{xcolor}
\usepackage{graphicx}
\usepackage{subcaption}
\usepackage{mathtools}
\usepackage{underscore}
\usepackage{nicematrix}

\NiceMatrixOptions{
    code-for-first-row = \ifmmode $\color{blue}\tiny\bfseries\mathversion{bold}$
        \else \bfseries\mathversion{bold}
        \fi,
    code-for-first-col = \ifmmode $\color{blue}\tiny\bfseries\mathversion{bold}$
        \else \bfseries\mathversion{bold}
        \fi,
    code-for-last-row = \ifmmode $\color{blue}\tiny\bfseries\mathversion{bold}$
        \else \bfseries\mathversion{bold}
        \fi,
    code-for-last-col = \ifmmode $\color{blue}\tiny\bfseries\mathversion{bold}$
        \else \bfseries\mathversion{bold}
        \fi
}

\usepackage{newtxtext,newtxmath}
\usepackage{enumitem}
\usepackage{hyperref}
\usepackage{crossreftools}

\makeatletter
\newcommand{\optionaldesc}[2]{%
  \phantomsection
  #1\protected@edef\@currentlabel{#1}\label{#2}%
}
\makeatother

\usepackage{tikz}
\usepackage{tikzit}
\usetikzlibrary{quantikz} %should be quantikz2 but EPTCS doesn't accept it
\tikzstyle{Large Black}=[fill=black, draw=black, shape=circle, tikzit category=GR, tikzit fill = black]
\tikzstyle{Large Empty}=[fill=white, draw=black, shape=circle, tikzit category=GR, tikzit fill =  white]
\tikzstyle{Large Empty Box}=[fill=white, draw=black, shape=rectangle, tikzit category=GR]
\tikzstyle{Small Black}=[fill=black, draw=black, shape=circle, inner sep=0pt, minimum size=6pt, tikzit category = GR]
\tikzstyle{Small Empty}=[fill=white, draw=black, shape=circle, inner sep=0pt, minimum size=6pt, tikzit category = GR]
\tikzstyle{Small Empty Box}=[fill=white, draw=black, shape=rectangle, inner sep=0pt, minimum size=8pt, tikzit category=GR]
\tikzstyle{Small Grey}=[fill={rgb,255: red,119; green,119; blue,119}, draw={rgb,255: red,119; green,119; blue,119}, shape=diamond, tikzit category=GR, tikzit fill={rgb,255: red,119; green,119; blue,119}, inner sep=0pt, minimum size=6pt]
\tikzstyle{Small Grey Empty}=[draw={rgb,255: red,119; green,119; blue,119}, shape=circle, tikzit category=GR, tikzit fill={rgb,255: red,119; green,119; blue,119}, inner sep=0pt, minimum size=6pt]
\tikzstyle{GreyTEXT}=[align=center,text={rgb,255: red,119; green,119; blue,119}]
\tikzstyle{Tiny Black}=[fill=black, draw=black, shape=circle, inner sep=0pt, minimum size=4pt, tikzit category = GR]

\tikzstyle{empty}=[-, fill={rgb,255: red,191; green,191; blue,191}, draw={rgb,255: red,191; green,191; blue,191}]
\tikzstyle{brace edge}=[-, tikzit draw=blue, decorate, decoration={brace,amplitude=1mm,raise=-1mm}]
\tikzstyle{grey}=[-, draw={rgb,255: red,119; green,119; blue,119}, densely dashed]
\tikzstyle{directed}=[->]
\graphicspath{ {./img/} }

\theoremstyle{definition}
\newtheorem{theorem}{Theorem}[section]
\newtheorem{corollary}[theorem]{Corollary}
\newtheorem{lemma}[theorem]{Lemma}
\newtheorem*{lemma*}{Lemma}
\newtheorem*{proposition*}{Proposition}
\newtheorem*{remark*}{Remark}

\newtheorem{definition}[theorem]{Definition}
\newtheorem{example}[theorem]{Example}

\usepackage{algorithm}
\usepackage[noend]{algpseudocode}

\newcommand{\RM}{\mathrm}
\newcommand{\counting}[1][]{\mathord{\#}\mathrm #1}

\newcommand{\problem}[1]{\mathbf{#1}}
\newcommand{\FlowSearch}{\problem{FlowSearch}}
\newcommand{\NonSing}{\problem{NonSing}}
\newcommand{\MaxRank}{\problem{MaxRank}}

\DeclareMathOperator{\rank}{rank}
\DeclareMathOperator{\rowrank}{row rank}

\usepackage{multicol}

  {\list{}{\leftmargin=#1\rightmargin=#1}\item[]}%
  {\endlist}

\DeclareCaptionFormat{cont}{#1 (continued)#2#3\par}

\title{Pauli Flow on Open Graphs with Unknown Measurement Labels}
\author{Piotr Mitosek
\institute{School of Computer Science, University of Birmingham}
\email{pbm148@student.bham.ac.uk}
}

\usepackage{etoolbox}
\apptocmd{\sloppy}{\hbadness 10000\relax}{}{}

\def\authorrunning{P. Mitosek}
\def\titlerunning{Pauli Flow on Open Graphs with Unknown Measurement Labels}

\begin{document}

\maketitle

\begin{abstract}
    One-way quantum computation, or measurement-based quantum computation, is a universal model of quantum computation alternative to the circuit model. The computation progresses by measurements of a pre-prepared resource state together with corrections of undesired outcomes via applications of Pauli gates to yet unmeasured qubits. The fundamental question of this model is determining whether computation can be implemented deterministically. Pauli flow is one of the most general structures guaranteeing determinism. It is also essential for polynomial time ancilla-free circuit extraction. It is known how to efficiently determine the existence of Pauli flow given an open graph together with a measurement labelling (a choice of measurements to be performed).

    In this work, we focus on the problem of deciding the existence of Pauli flow for a given open graph when the measurement labelling is unknown. We show that this problem is in $\RM{RP}$ by providing a random polynomial time algorithm. To do it, we extend previous algebraic interpretations of Pauli flow, by showing that, in the case of $X$ and $Z$ measurements only, flow existence corresponds to the right-invertibility of a matrix derived from the adjacency matrix. We also use this interpretation to show that the number of output qubits can always be reduced to match the number of input qubits while preserving the existence of flow.
\end{abstract}

\section{Introduction}
One-way quantum computation, or measurement-based quantum computation (MBQC), is a quantum computation model alternative to the circuit model, yet at least as powerful \cite{raussendorfOneWayQuantumComputer2001, raussendorfOnewayQuantumComputer2002, raussendorfMeasurementbasedQuantumComputation2003}. The computation is performed entirely by a series of measurements on a pre-prepared resource state. MBQC is expected to become critical for quantum communication applications. For instance, MBQC is used in blind quantum computation \cite{broadbentUniversalBlindQuantum2009}, a quantum protocol allowing a client to outsource part of the computation to a server without compromising the client's privacy. However, the quantum measurements are random by their nature, which means that in each step two different outcomes may be obtained, where only one is the desired outcome. Therefore, it is the fundamental problem of MBQC to determine whether the desired computation can be performed deterministically.

A structure essential for deterministic MBQC is flow defined on a labelled open graph -- a graph state with designated input and output qubits together with a choice of measurements for non-outputs. The computation is performed by a series of measurements of non-outputs in one of six settings: planar $XY$, $XZ$, and $YZ$ measurements and Pauli $X$, $Y$, and $Z$ measurements. After each measurement, a correction might be applied which is intended to fix a potential undesired outcome, bringing the open graph into a state equivalent to one that would be reached if a desired outcome was observed. The corrections are performed by applying Pauli gates to yet unmeasured qubits. One of the notions of determinism for such computation corresponds to the existence of flow structure in the labelled open graph. There are many different types of flow, ranging from the causal flow \cite{danosDeterminismOnewayModel2006}, through generalized flow \cite{browneGeneralizedFlowDeterminism2007}, all the way to the Pauli flow \cite{browneGeneralizedFlowDeterminism2007} and its extended variant \cite{mhallaCharacterisingDeterminismMBQCs2022}. Every causal flow is also a generalized flow. Every generalized flow is also a Pauli flow. However, the reverses are not true, thus making Pauli flow the most universal.

The flow is also crucial when working with ZX calculus \cite{coeckeInteractingQuantumObservables2011, vandeweteringZXcalculusWorkingQuantum2020a}, a graphical calculus for quantum computation, where the notions of flow correspond to the most general subclasses of diagrams for which a polynomial time ancilla-free circuit extraction algorithm exists \cite{backensThereBackAgain2021, simmonsRelatingMeasurementPatterns2021}. While diagrams from graphical calculi are easier to work with than circuits, they are not a notion of computation understood by quantum computers. For that, a corresponding circuit must be found, however, circuit extraction in general is $\counting[P]$-hard \cite{debeaudrapCircuitExtractionZXDiagrams2022c}. In \cite{duncanGraphtheoreticSimplificationQuantum2020}, circuits are optimized by translating them to ZX, simplifying diagrams, and translating them back to circuits. Other versions of this approach have appeared since for example in \cite{staudacherReducing2QuBitGate2023}. There, diagram transformations necessarily must preserve flow, which was a subject of for instance \cite{mcelvanneyCompleteFlowPreservingRewrite2023, mcelvanneyFlowpreservingZXcalculusRewrite2023a}.

While the existence of different flow variants can be efficiently determined for a given labelled open graph \cite{debeaudrapFindingFlowsOneway2008,debeaudrapCompleteAlgorithmFind2007,mhallaFindingOptimalFlows2008a,backensThereBackAgain2021,simmonsRelatingMeasurementPatterns2021}, many questions remain open. For instance, what are other properties of labelled open graphs with flow? Previous works looking at the necessary properties the graph must satisfy to contain flow include \cite{mhallaWhichGraphStates2014a, markhamEntanglementFlowClassical2014, simmonsRelatingMeasurementPatterns2021}. We contribute to the general picture of what flow is and which graphs exhibit it.

\paragraph{Results overview} In this work, we show that given an (unlabelled) open graph it is possible to efficiently determine whether there exists a choice of measurements resulting in flow. In particular, we show the following problem is in $\RM{RP}$, meaning that there exists a random polynomial-time algorithm solving it.

\begin{quote}
    $\FlowSearch$\\
    \textbf{Input:} An open graph $(G, I, O)$.\\
    \textbf{Output:} $True$ if there exists a measurement labelling $\lambda$ such that the labelled open graph $(G, I, O, \lambda)$ exhibits Pauli flow, and $False$ otherwise.
\end{quote}

To show that, we adapt the algebraic interpretation of flow from \cite{mhallaWhichGraphStates2014a}. In that paper, the authors relate the existence of generalised flow for $XY$ planar measurements only to the existence of the right inverse of a particular matrix satisfying additional requirements. We relate Pauli flow in the case of only $X$ and $Z$ measurement to just right invertibility. Then, we use the algebraic interpretation to reduce $\FlowSearch$ to $\MaxRank$, a problem about maximizing the rank of a matrix whose entries can contain variables. Our approach is limited to finding labelling consisting of Pauli $X$ and $Z$ measurements. The corresponding measurement scheme is Clifford. Thus, the current algorithm may not be ideal for searching more complex labels which make the corresponding labelled open graph have Pauli flow. However, as we will prove, the graph having Pauli flow for some measurement labels also has flow for labels consisting of only Pauli $X$ and $Z$ measurements. In other words, if a graph does not have flow for some label consisting of just $X$ and $Z$ measurements then it cannot have flow for any label. Therefore, our approach is sufficient to solve $\FlowSearch$.

Further, we use the algebraic interpretation of flow to explore the conditions on the sets of inputs and outputs necessary to have flow. In particular, we prove that if a labelled open graph has Pauli flow, then some outputs can be removed and the resulting open graph will still have Pauli flow for some measurement labelling.

\paragraph{Structure}
In section \ref{background}, we present the background. Section \ref{main section} is the main part of the paper. There, we formally describe our results and prove them. Finally, in section \ref{conclusions and further work}, we discuss the conclusions and possible further work. We include some technical proofs and additional examples in the appendices.

\section{Background}\label{background}
In this section, we first look at the general concept of measurement-based quantum computation, where we define Pauli flow. Next, we explain computational complexity definitions critical in our proofs.

\subsection{Measurement-Based Quantum Computation}
As outlined earlier, we focus on the version of MBQC where measurements are performed on open graphs and the allowed measurements are single-qubit: Pauli measurements or planar measurements from $XY$, $YZ$, or $XZ$ planes. The description of the computation is given in the form of measurement patterns. The labelled open graphs can be viewed as runnable measurement patterns without corrections \cite{backensThereBackAgain2021}.

\begin{definition}[Open Graph]
    An \textbf{open graph} is a triple $(G,I,O)$ consisting of:\begin{itemize}[noitemsep]
        \item an undirected graph $G = (V,E)$,
        \item a set of \textbf{inputs} $I \subseteq V$,
        \item a set of \textbf{outputs} $O \subseteq V$.
    \end{itemize}
    We define the sets of \textbf{non-inputs} $\overline{I} := V \setminus I$, \textbf{non-outputs} $\overline{O} := V \setminus O$ and \textbf{internal vertices} $B := V\setminus (I\cup O) = \overline{I} \cap \overline{O}$.

    An \textbf{odd neighbourhood} of a set of vertices $A$ is denoted $Odd(A)$ and consists of all vertices neighbouring an odd number of elements of $A$: $Odd(A) := \{ v \in V \mid \#\{ a \in A \mid va \in E \} \text{ is odd}\}$.
\end{definition}

\begin{definition}
    A \textbf{measurement labelling} for an open graph $(G,I,O)$ is any function $\lambda$ sending non-outputs $\overline{O}$ to labels $\{ X, Y, Z, XZ, XY, YZ \}$, satisfying $\lambda(v) \in \{ X, XY, Y \}$ for all $v \in I \setminus O$. A \textbf{labelled open graph} is a quadruple $(G,I,O,\lambda)$ where $(G,I,O)$ is an open graph and $\lambda$ is a measurement labelling.
\end{definition}

An open graph is prepared by entangling input qubits, corresponding to the vertices in $I$, with qubits prepared in $\ket{+}$ corresponding to the vertices in $\overline{I}$, by applying $CZ$ gate between pair of qubits corresponding to each edge in the graph. Since the $CZ$ gates commute, the order of $CZ$ applications does not matter. When the set of inputs is empty, the notion of an open graph collapses to the graph state. Next, qubits corresponding to the elements of $\overline{O}$ are measured according to the measurement labelling. See figure \ref{fig:ex233} for an example.

The computation given in the form of the measurement pattern depends on the outcomes of the measurements. After each measurement, an undesired outcome may occur and it must be corrected. The corrections are performed by applying Pauli $X$ and $Z$ gates to the yet unmeasured qubits. The result of a chosen measurement outcome, together with corrections, is called a branch. There are many different notions of determinism \cite{browneGeneralizedFlowDeterminism2007, mhallaWhichGraphStates2014a}. The one we focus on is the strong, uniform, and stepwise determinism -- all branches of the computation are equal up to a global phase, for any choice of measurement angles, and the intermediate patterns after performing a subset of measurements also have these properties \cite{browneGeneralizedFlowDeterminism2007, mhallaWhichGraphStates2014a, backensThereBackAgain2021}. The measurement pattern is strongly, uniformly, and stepwise deterministic when the corresponding labelled open graph has Pauli flow \cite{browneGeneralizedFlowDeterminism2007}. We define the notion of Pauli flow (based on \cite{simmonsRelatingMeasurementPatterns2021}).

\begin{definition}[Pauli flow]
    A \textbf{Pauli flow} for a labelled open graph $(G,I,O,\lambda)$ is a pair $(c,\prec)$ where $c$ is a function $\overline{O} \to \mathcal{P}(\overline{I})$ and $\prec$ is a strict partial order on $\overline{O}$, such that for all $u \in \overline{O}$:\begin{enumerate}[noitemsep]
        \item[{\crtcrossreflabel{(P1)}[P1]}] $\forall v \in c(u) . u \ne v \wedge \lambda(v) \notin \{ X, Y \} \Rightarrow u \prec v$
        \item[{\crtcrossreflabel{(P2)}[P2]}] $\forall v \in Odd(c(u)) . u \ne v \wedge \lambda(v) \notin \{ Y, Z \} \Rightarrow u \prec v$
        \item[{\crtcrossreflabel{(P3)}[P3]}] $\forall v \in \overline{O} . \neg (u \prec v) \wedge u \ne v \wedge \lambda(v) = Y \Rightarrow \left( v \in c(u) \Leftrightarrow v \in Odd(c(u)) \right)$
        \item[{\crtcrossreflabel{(P4)}[P4]}] $\lambda(u) = XY \Rightarrow u \notin c(u) \wedge u \in Odd(c(u))$
        \item[{\crtcrossreflabel{(P5)}[P5]}] $\lambda(u) = XZ \Rightarrow u \in c(u) \wedge u \in Odd(c(u))$
        \item[{\crtcrossreflabel{(P6)}[P6]}] $\lambda(u) = YZ \Rightarrow u \in c(u) \wedge u \notin Odd(c(u))$
        \item[{\crtcrossreflabel{(P7)}[P7]}] $\lambda(u) = X \Rightarrow u \in Odd(c(u))$
        \item[{\crtcrossreflabel{(P8)}[P8]}] $\lambda(u) = Z \Rightarrow u \in c(u)$
        \item[{\crtcrossreflabel{(P9)}[P9]}] $\lambda(u) = Y \Rightarrow \left( u \in c(u) \oplus u \in Odd(c(u)) \right)$, where $\oplus$ stands for XOR.
    \end{enumerate}
    We call $c$ the \textbf{correction function} and sets $c(v)$ for $v \in \overline{O}$ are called the \textbf{correction sets}.
\end{definition}

\begin{example}\label{flow def example}
    Consider open graphs from figure \ref{fig:ex233}. The graph from \ref{fig:ex233 no flow} does not have Pauli flow. On the other hand, the graph from \ref{fig:ex233 1z flow} has Pauli flow. For instance, taking $D \prec A$ and $D \prec B$ with the following correction function result in the flow:
    \begin{align*}
        c(A) &= \{ C, E \},\quad &Odd(c(A)) &= \{ A, F, G, H \},
        &c(B) &= \{ E \},\quad &Odd(c(B)) &= \{ B, F \},\\
        c(C) &= \{ G \},\quad &Odd(c(C)) &= \{ C, D \},
        &c(D) &= \{ D \},\quad &Odd(c(D)) &= \{ A, B, G, H \},\\
        c(E) &= \{ F \},\quad &Odd(c(E)) &= \{ E \}.
    \end{align*}
\end{example}

In MBQC, the meaning of the Pauli flow is as follows. The prepared open graph state is measured according to the partial order. When a measurement of a vertex $u$ results in an undesired outcome, then the measurement error can be fixed by applying the $X$ gate to all vertices in $c(u)$. The flow guarantees that each correction is physically possible and that every error can be corrected independently of the outcomes of the previous measurements. In other words, the MBQC becomes deterministic. For a detailed explanation, see \cite{simmonsRelatingMeasurementPatterns2021}.

\begin{figure}
    \centering
    \begin{subfigure}[b]{0.4\textwidth}
        \centering
        $\scalebox{1}{\begin{tikzpicture}
	\begin{pgfonlayer}{nodelayer}
		\node [style=Small Empty Box] (0) at (-4, 1.5) {};
		\node [style=Small Empty Box] (1) at (-4, -1.5) {};
		\node [style=Small Black] (2) at (0, 3) {};
		\node [style=Small Black] (3) at (0, 0) {};
		\node [style=Small Black] (4) at (0, -3) {};
		\node [style=Small Empty] (5) at (4, 3) {};
		\node [style=Small Empty] (6) at (4, 0) {};
		\node [style=Small Empty] (7) at (4, -3) {};
		\node [style=none] (8) at (0, 4.5) {};
		\node [style=none] (9) at (0, -4) {};
		\node [style=none] (10) at (6, 0) {};
		\node [style=none] (11) at (-6, 0) {};
		\node [style=none] (12) at (-4, 2.25) {$A$};
		\node [style=none] (13) at (-4, -0.75) {$B$};
		\node [style=none] (14) at (0, 3.75) {$C$};
		\node [style=none] (15) at (0, 0.75) {$D$};
		\node [style=none] (16) at (0, -2.25) {$E$};
		\node [style=none] (17) at (4, 3.75) {$F$};
		\node [style=none] (18) at (4, 0.75) {$G$};
		\node [style=none] (19) at (4, -2.25) {$H$};
		\node [style=Small Black] (20) at (0, -3) {};
		\node [style=Small Black] (21) at (0, 3) {};
		\node [style=Small Black] (22) at (-4, 1.5) {};
		\node [style=Small Black] (23) at (-4, -1.5) {};
	\end{pgfonlayer}
	\begin{pgfonlayer}{edgelayer}
		\draw (0) to (2);
		\draw (0) to (3);
		\draw (2) to (1);
		\draw (3) to (1);
		\draw (1) to (4);
		\draw (4) to (5);
		\draw (2) to (6);
		\draw (3) to (6);
		\draw (3) to (7);
		\draw (7) to (2);
	\end{pgfonlayer}
\end{tikzpicture}}$
        \caption{A labelled open graph with two inputs and three outputs. All non-outputs are $X$ labelled.}
        \label{fig:ex233 no flow}
    \end{subfigure}
    \hfill
    \begin{subfigure}[b]{0.4\textwidth}
        \centering
        $\scalebox{1}{\begin{tikzpicture}
	\begin{pgfonlayer}{nodelayer}
		\node [style=Small Empty Box] (0) at (-4, 1.5) {};
		\node [style=Small Empty Box] (1) at (-4, -1.5) {};
		\node [style=Small Black] (2) at (0, 3) {};
		\node [style=Small Grey] (3) at (0, 0) {};
		\node [style=Small Black] (4) at (0, -3) {};
		\node [style=Small Empty] (5) at (4, 3) {};
		\node [style=Small Empty] (6) at (4, 0) {};
		\node [style=Small Empty] (7) at (4, -3) {};
		\node [style=none] (8) at (0, 4.5) {};
		\node [style=none] (9) at (0, -4) {};
		\node [style=none] (10) at (6, 0) {};
		\node [style=none] (11) at (-6, 0) {};
		\node [style=none] (12) at (-4, 2.25) {$A$};
		\node [style=none] (13) at (-4, -0.75) {$B$};
		\node [style=none] (14) at (0, 3.75) {$C$};
		\node [style=GreyTEXT] (15) at (0, 0.75) {$D$};
		\node [style=none] (16) at (0, -2.25) {$E$};
		\node [style=none] (17) at (4, 3.75) {$F$};
		\node [style=none] (18) at (4, 0.75) {$G$};
		\node [style=none] (19) at (4, -2.25) {$H$};
		\node [style=Small Black] (20) at (0, -3) {};
		\node [style=Small Black] (21) at (0, 3) {};
		\node [style=Small Black] (22) at (-4, 1.5) {};
		\node [style=Small Black] (23) at (-4, -1.5) {};
	\end{pgfonlayer}
	\begin{pgfonlayer}{edgelayer}
		\draw (0) to (2);
		\draw [style=grey] (0) to (3);
		\draw (2) to (1);
		\draw [style=grey] (3) to (1);
		\draw (1) to (4);
		\draw (4) to (5);
		\draw (2) to (6);
		\draw [style=grey] (3) to (6);
		\draw [style=grey] (3) to (7);
		\draw (7) to (2);
	\end{pgfonlayer}
\end{tikzpicture}}$
        \caption{The labelled open graph from \ref{fig:ex233 no flow}, but with vertex $D$ labelled $Z$.}
        \label{fig:ex233 1z flow}
    \end{subfigure}

    \caption{Examples of labelled open graphs. Outputs are denoted by empty circles, $X$ labelled non-outputs by filled circles and inputs by a square. The $Z$ labelled vertices are denoted with grey diamonds and edges including them are dashed and grey.}
    \label{fig:ex233}
\end{figure}
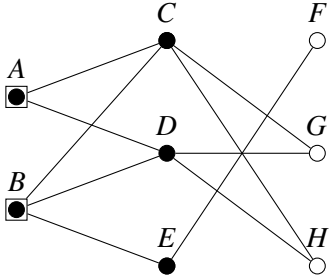
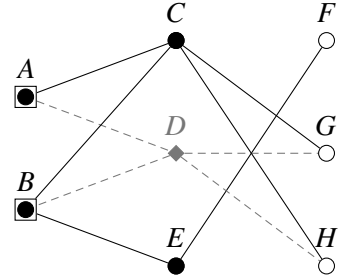

\subsection{Computational Complexity}
We assume familiarity with standard complexity terminology. We are mainly concerned about the \textbf{random polynomial time} class $\RM{RP}$, where randomness is permitted with one-sided bounded error:

\begin{definition}
    The class of $\RM{RP}$ consists of problems $A$ solvable by a non-deterministic Turing Machine $M$ that given input $a$ proceeds as follows:\begin{itemize}[noitemsep]
        \item if $a$ is a ``NO'' instance, then $M$ always rejects $a$,
        \item if $a$ is a ``YES'' instance, then $M$ accepts $a$ on at least half of its computation paths.
    \end{itemize}
\end{definition}

For a problem in $\RM{RP}$ there exists a polytime algorithm with random number generator access that for ``NO'' instances always returns ``NO'' and for ``YES'' instances it returns ``YES'' with probability at least $\frac{1}{2}$ and otherwise it returns ``NO''. This last case is the error of the computation and happens with the probability smaller than $\frac{1}{2}$. By running the algorithm multiple times it is possible to reduce the error probability -- $50$ runs results in error probability below $\frac{1}{2^{50}}$ which is sufficient for all practical applications.

The following is the essential problem in our work. The definition is adapted from \cite{bussComputationalComplexityProblems1999} (in contrast, we do not require the input matrix to be square).

\begin{quote}
    $\MaxRank$\\
    \textbf{Fixed:} A commutative ring $R$, and subsets $E, S \subseteq R$ of entries and solutions.\\
    \textbf{Input:} Natural numbers $m, n, t, r$ and $m \times n$ matrix $M$ with entries from $E \cup \{ x_1, \dots, x_t \}$.\\
    \textbf{Output:} $True$ if $\rank M(a_1, \dots, a_t) \ge r$ for some $a_1 \dots a_t \in S^t$, and $False$ otherwise.
\end{quote}

$M(a_1,\dots,a_t)$ stands for the matrix with substituted variables $x_1 \mapsto a_1, \dots, x_t \mapsto a_t$. We skip the specification of $m,n,t$, as these numbers are explicit from the input matrix. Thus, we will write instances of $\MaxRank$ as pairs $(M,r)$ of the matrix and the desired minimal rank under some valuation.

In \cite{bussComputationalComplexityProblems1999}, it is shown that when $R$ is a finite field and each variable occurs at most once, the problem is in $\RM{RP}$. We extend their approach to show that the problem is in $\RM{RP}$ also when each variable appears in at most one row or one column. For that, we need the notion of multi-affine polynomials (adapted from \cite{bussComputationalComplexityProblems1999}). In general, the $\MaxRank$ problem with $R$ being a finite field is $\RM{NP}$-complete \cite{bussComputationalComplexityProblems1999}.

\begin{definition}
    A multivariable polynomial is \textbf{multi-affine} when each variable has a degree at most one.
\end{definition}

By extending the proof \cite[Theorem 28]{bussComputationalComplexityProblems1999}, we get the following version. Here, the entries of the matrix are allowed to be given by multi-affine expressions over $E \cup \{ x_1, \dots, x_t \}$, not just elements of the set.

\begin{theorem}\label{main coco result}
    The following version of $\MaxRank$ is in $\RM{RP}$:\begin{itemize}[noitemsep]
        \item $R = E = S = \mathbb{F}_s$ is a finite field,
        \item each variable appears in at most one row or at most one column\footnote{That is if a variable appears in entries at positions $(a_1,b_1), (a_2,b_2), (a_3, b_3), \dots$, then either $a_1 = a_2 = a_3 = \dots$ or $b_1 = b_2 = b_3 = \dots$.},
        \item the entries are given by (polynomially long) multi-affine expressions over $E \cup \{ x_1, \dots, x_t \}$.
    \end{itemize}
\end{theorem}

We prove the above theorem in the appendix \ref{more about maxrank}, where we also give an example and talk about other known results for $\MaxRank$ and similar problems.

\section{Finding labelling resulting in Pauli flow}\label{main section}
We start by formally defining theorems capturing our results. We show that given an open graph $(G, I, O)$, there exists a random polytime algorithm deciding the existence of measurement labelling $\lambda$ such that the labelled open graph $(G, I, O, \lambda)$ has Pauli flow, i.e we show that $\FlowSearch$ defined in the introduction is in $\RM{RP}$.

\begin{theorem}\label{main result}
    $\FlowSearch$ is in $\RM{RP}$.
\end{theorem}

\begin{example}
    Consider the open graph from figure \ref{fig:ex233}, ignoring the labelling. When considered as an instance of $\FlowSearch$, the answer is $True$, because the labelling from figure \ref{fig:ex233 1z flow} results in Pauli flow.
\end{example}

To show the above theorem, we expand on the known \cite{mhallaWhichGraphStates2014a} correspondence between flow and certain algebraic properties of various matrices. We also show that given an open graph with more outputs than inputs, it is always possible to reduce the number of outputs to match the number of inputs, while preserving the flow.

\begin{theorem}\label{removing outputs}
    Suppose $(G,I,O,\lambda)$ has Pauli flow and $|O|>|I|$. Then there exists a subset $O' \subseteq O$ such that $|O'| = |I|$ and a labelling $\lambda'$ such that $(G,I,O',\lambda')$ has Pauli flow.
\end{theorem}

Labelled open graphs with Pauli flow and equal numbers of inputs and outputs have multiple elegant properties. Firstly, they correspond to a circuit without ancillas and hence a unitary. Further, the graph and flow can be ``reversed'' (\cite[Theorem 3.4]{mhallaWhichGraphStates2014a} for generalized flow theorem). Next, there is a unique correction function in the focussed Pauli flow. In the case of $X$ and $Z$ labels only, the existence of such a correction function corresponds to the invertibility of what we call the flow matrix, rather than just right-invertibility. Finding the inverse can be implemented with faster algorithms than Gaussian elimination. Thus, it is useful to transform labeled open graphs with more outputs than inputs in order to equalise both sizes.

The remainder of this section are the proofs of the theorems \ref{main result} and \ref{removing outputs}. The first proof is divided into three subsections. In the first one, we do preliminary simplifications of the $\FlowSearch$ problem. The second subsection presents an algebraic interpretation of the Pauli flow, by giving a correspondence between Pauli flow and matrix invertibility. Finally, in the third subsection, we show how $\FlowSearch$ instances can be transformed into $\MaxRank$ instances, and that those instances can be solved by a random polytime algorithm. The proof of theorem \ref{removing outputs} is presented in the final subsection.

\subsection{Reducing measurement options}
We start by reducing the number of possible options for the measurement basis to just Pauli measurements.

\begin{theorem}\label{red to x y z}
    Suppose that a labelled open graph $(G,I,O,\lambda)$ has Pauli flow. Then, there exists $\lambda' \colon \overline{O} \to \{ X, Y, Z \}$ such that $(G,I,O,\lambda')$ has Pauli flow.
\end{theorem}
\begin{proof}
    We start by fixing a Pauli flow $(c,\prec)$ on $(G,I,O,\lambda)$. The conditions for a planar measurement $XY$ combine the requirements for the two Pauli measurements $X$ and $Y$. Hence, swapping $XY$ measurements to $X$ preserves $(c,\prec)$ as the Pauli flow. Similarly, we can swap $XZ$ to $X$ and $YZ$ to $Z$.
\end{proof}

There can be many different Pauli flows for a single labelled open graph. However, there exists a special type of Pauli flow known as the \textbf{focussed Pauli flow}. The following definition is based on \cite[Definition 4.3]{simmonsRelatingMeasurementPatterns2021}.

\begin{definition}[Focussed Pauli flow]
    The Pauli flow $(c,\prec)$ is \textbf{focussed} when for all $v \in \overline{O}$ the following hold:\begin{enumerate}[noitemsep]
        \item[{\crtcrossreflabel{(F1)}[F1]}] $\forall w \in (\overline{O} \setminus \{ v \}) \cap c(v) . \lambda(w) \in \{ XY, X, Y \}$
        \item[{\crtcrossreflabel{(F2)}[F2]}] $\forall w \in (\overline{O} \setminus \{ v \}) \cap Odd(c(v)) . \lambda(w) \in \{ XZ, YZ, Y, Z \}$
        \item[{\crtcrossreflabel{(F3)}[F3]}] $\forall w \in (\overline{O} \setminus \{ v \}) . \lambda(w) = Y \Rightarrow (w \in c(v) \Leftrightarrow w \in Odd(c(v)))$
    \end{enumerate}
\end{definition}

\begin{example}
    Consider the Pauli flow from the example \ref{flow def example} -- it is not focussed as $Odd(c(D))$ contains $X$ labelled vertices $A$ and $B$. Changing $c(D)$ to $\{ C, D \}$ results in $Odd(c(D)) = \emptyset$ and the flow becomes focussed.
\end{example}

Importantly, the existence of flow is equivalent to the existence of focussed Pauli flow, as captured by the following theorem \cite[Lemma 4.6]{simmonsRelatingMeasurementPatterns2021}:

\begin{theorem}[{\cite[Lemma 4.6]{simmonsRelatingMeasurementPatterns2021}}]\label{foc flow exists}
    For any open labelled graph, if a Pauli flow exists then there also exists a focussed Pauli flow.
\end{theorem}

In the case of only Pauli measurements, the conditions from the focussed flow mean that the corrector sets can only consist of $X$ and $Y$ measured vertices and the corrected vertex, while the odd neighbourhoods of the corrector sets can only consist of the $Z$ and $Y$ measured vertices and the corrected vertex.

It is easier to search for the focussed Pauli flow, as the notion of the focussed flow in the case of Pauli bases does not require partial order as captured below. The proof follows from \cite[Lemma B.11]{simmonsRelatingMeasurementPatterns2021}.

\begin{lemma}\label{no order}
    If $(c,\prec)$ is a focussed Pauli flow for a labelled open graph $(G,I,O,\lambda)$ with only Pauli measurements, then so is $(c,\emptyset)$.
\end{lemma}

Further, we can improve theorem \ref{red to x y z} to only look for $X$ and $Z$ labelling, getting rid of the $Y$ measurements:

\begin{theorem}\label{red to x z}
    Suppose that a labelled open graph $(G,I,O,\lambda)$ has Pauli flow. Then, there exists $\lambda' \colon \overline{O} \to \{ X, Z \}$ such that $(G,I,O,\lambda')$ has Pauli flow.
\end{theorem}
\begin{proof}
    Let $(G,I,O,\lambda)$ have Pauli flow. By theorem \ref{red to x y z}, there is $\lambda_1 \colon \overline{O} \to \{ X, Y, Z \}$ such that $(G,I,O,\lambda_1)$ has Pauli flow. By theorem \ref{foc flow exists} and lemma \ref{no order}, $(G,I,O,\lambda_1)$ has some focussed Pauli flow $(c,\emptyset)$. If there is no $u$ with $\lambda_1(u) = Y$, the thesis follows. Otherwise, consider such $u$. Define $\lambda_2 \colon \overline{O} \to \{ X, Y, Z \}$ as follows. $\lambda_2(v) = \lambda_1(v)$ for $v \ne u$. By \ref{P9}, $u \in p(u) \oplus u \in Odd(p(u))$. If $u \in p(u)$, we set $\lambda_2(u) = Z$ and when $u \in Odd(p(u))$, we set $\lambda_2(u) = X$. Then $(G,I,O,\lambda_2)$ has Pauli flow $(c,\prec)$, where $c$ is the same correction function as for $(G,I,O,\lambda_1)$ and $\prec$ is given by $v \prec u$ for all $v \in \overline{O}\setminus \{ u \}$, i.e.\ $u$ is the last vertex in this order. We verify $(c,\prec)$ indeed is the Pauli flow. Conditions \ref{P4}, \ref{P5}, \ref{P6} hold automatically. Similarly, \ref{P7}, \ref{P8}, \ref{P9} hold for vertices other than $u$. The choice for the new label of $u$ is done by using \ref{P9} for the old flow, thus ensuring that $c$ works for correction of $u$ with new label. Hence, only \ref{P1}, \ref{P2} and \ref{P3} remain. By construction, $u$ is last in the order, hence if $u \in c(v)$ or $u \in Odd(c(v))$, the necessary order condition is guaranteed to hold. When $v \in c(u)$ or $v \in Odd(c(u))$, then the $\lambda_2(v)=\lambda_1(v)$ ensures that no new order requirement appears. Similarly, $\emptyset$ order worked for other pairs of vertices. Thus, \ref{P1}, \ref{P2} and \ref{P3} all hold. By focussing the flow and relabelling one vertex labelled with $Y$ at a time, we can get rid of all $Y$ measured vertices, ending the proof.
\end{proof}

Thanks to the above lemmata, instead of looking for $\lambda$ resulting in Pauli flow $(c,\prec)$, we can just look for $\lambda$ into Pauli measurements resulting in a focussed Pauli flow $(c,\emptyset)$.

Finally, we can omit cases with $I \cap O \ne \emptyset$, by reducing those to have $I \cap O = \emptyset$.

\begin{theorem}\label{no i cap o}
    Let $(G,I,O,\lambda)$ be a labelled open graph. Then $(G,I,O,\lambda)$ has Pauli flow if and only if $(G',I\setminus O, O\setminus I,\lambda)$ does, where $G'$ is $G$ with vertices in $I \cap O$ removed.
\end{theorem}
\begin{proof}
    Suppose $v \in I \cap O$. A correction function $c$ in Pauli flow has codomain $\mathcal{P}(\overline{I})$ and $v \in I$ so $v$ cannot be used in any correction set. Further, as $v \in O$, $v$ is not measured and so $c(v)$ is undefined and the partial order does not operate on $v$. Therefore $v$ does not impact any of the nine flow conditions in any way.
\end{proof}

Combining the above simplifications, we can restrict $\FlowSearch$ problem to cases $(G,I,O)$ with $I \cap O = \emptyset$, where we are looking for $\lambda \colon \overline{O} \to \{ X, Z \}$ such that $(G,I,O,\lambda)$ has focussed Pauli flow.

Since the removal of $Z$-measured vertices preserves Pauli flow (which we mention later), it is also possible to view the above problem as follows: given $(G,I,O)$, is there an induced open subgraph with the same set of inputs and outputs, that has Pauli flow with only $X$ labels?

\subsection{Algebraic interpretation of flow}
Unless specified otherwise, from now on, we only consider $(G,I,O)$ with $I \cap O = \emptyset$. The key construction used to translate the $\FlowSearch$ problem into a linear algebra problem is the reduced adjacency matrix \cite{mhallaWhichGraphStates2014a} (note, that in that paper the word \textit{induced} is used in place of \textit{reduced}).

\begin{definition}
    Let $(G,I,O)$ be an open graph. We define the \textbf{adjacency matrix} $A_G$ over $\mathbb{F}_2$ as a $|V| \times |V|$ matrix with ${\left(A_{G}\right)_{u,v}} = 1$ when $uv \in E$ and $0$ otherwise. The \textbf{reduced adjacency matrix} ${A_G \mid}_{\overline{I}}^{\overline{O}}$ is the $|\overline{O}| \times |\overline{I}|$ minor of the adjacency matrix of $G$. The minor is obtained by removing the outputs' rows and the inputs' columns.
\end{definition}

The key property of the reduced adjacency matrix linking it to the Pauli flow is right-invertibility. In \cite{mhallaWhichGraphStates2014a}, a version lining $XY$ measurements only to the right-invertibility was established. Subsequently, this theorem was extended by Miriam Backens to work for both $X$ and $XY$ measurements:

\begin{theorem}\label{red adj matrix invertibility with DAG}
    Let $\mathcal{G} = (G,I,O,\lambda)$ be a labelled open graph with $\lambda(v) \in \{ X, XY \}$ for all $v \in \overline{O}$. Then $\mathcal{G}$ has focussed Pauli flow if and only if there exists a directed graph $F = (V,E_F)$ satisfying the following two properties:\begin{itemize}[noitemsep]
        \item Let $E_F' = \{(u,v) \in E_F \mid \lambda(u) = XY \}$, then the subgraph $F' = (V,E_F')$ is acyclic,
        \item $A_G\mid_{\overline{I}}^{\overline{O}} \cdot A_F\mid_{\overline{O}}^{\overline{I}} = Id_{\overline{O}}$, where $A_F\mid_{\overline{O}}^{\overline{I}}$ is the $\overline{I} \times \overline{O}$ minor of $A_F$ obtained by removing inputs' rows and outputs' columns.
    \end{itemize}
    Further, the columns of $A_F\mid_{\overline{O}}^{\overline{I}}$ encode the correction sets of the vertices in $\overline{O}$.
\end{theorem}

\iffalse
The proof has not yet been published. We include it in the appendix \ref{red adj matrix invertibility with DAG proof} for clarity. Note, that in the case of $X$ measurements only, the graph $F'$ in the theorem \ref{red adj matrix invertibility with DAG} is empty and hence it is automatically acyclic. Hence, the condition simplifies to just the existence of the right inverse.
\fi

The proof has not yet been published \cite{mitosekbackensupcoming}. Note, that in the case of $X$ measurements only, the graph $F'$ in the theorem \ref{red adj matrix invertibility with DAG} is empty and hence it is automatically acyclic. Hence, the condition simplifies to just the existence of the right inverse.

\begin{corollary}
    \label{red adj matrix invertibility}
    Let $\mathcal{G} = (G,I,O,\lambda)$ be a labelled open graph with $\lambda(v) = X$ for all $v \in \overline{O}$. Then $\mathcal{G}$ has focussed Pauli flow if and only if ${A_G \mid}_{\overline{I}}^{\overline{O}}$ is right-invertible over $\mathbb{F}_2$. Further, the columns of a potential right inverse of $A_G\mid_{\overline{I}}^{\overline{O}}$ encode the correction sets of the vertices in $\overline{O}$.
\end{corollary}

See figure \ref{fig:ex233 no flow mat} for an example of using corollary \ref{red adj matrix invertibility} and figure \ref{fig:ex x xy full} in appendix \ref{red adj matrix invertibility with DAG proof} for an example of using theorem \ref{red adj matrix invertibility with DAG}.

The $Z$ labelled vertices can be removed and introduced without affecting flow existence, as captured by the following lemmata. In these, $G[A]$ stands for the subgraph of $G$ induced by vertices in $A$.

\begin{lemma}[Removal of $Z$ measured vertex {\cite[Lemma D.6]{simmonsRelatingMeasurementPatterns2021}}]\label{Z removal}
    Let $(G,I,O,\lambda)$ be a labelled open graph with Pauli flow and with $\lambda(v) = Z$ for some $v \in \overline{O}$. Then $v$ can be removed without affecting flow existence. In other words, $(G[V \setminus \{ v \}],I,O,\lambda|_{\overline{O} \setminus \{v \}})$ has Pauli flow.
\end{lemma}

\begin{lemma}[Introduction of $Z$ measured vertex {\cite[Proposition 4.1]{mcelvanneyCompleteFlowPreservingRewrite2023}}]\label{Z intro}
    Let $(G,I,O,\lambda)$ have Pauli flow. Then a new $Z$ measured vertex $v \notin V$ can be added to $G$, with any edges from $v$, without affecting flow existence. In other words, any labelled open graph $(G',I,O,\lambda')$ has Pauli flow, where:
    \begin{gather*}
        V(G') = V \cup \{ v \} {\hskip 0.15\textwidth} G'[V] = G {\hskip 0.15\textwidth} \lambda'|_{V \setminus O} = \lambda {\hskip 0.15\textwidth} \lambda'(v) = Z
    \end{gather*}
\end{lemma}

The proof for the removal appears in \cite[Lemma D.6]{simmonsRelatingMeasurementPatterns2021}, and is a reformulation of \cite[Lemma 4.7]{backensThereBackAgain2021}.

Now, we extend the corollary \ref{red adj matrix invertibility} to also capture $Z$ measurements, by considering a slightly different matrix that we call the flow matrix.

\begin{definition}[Flow matrix]\label{flow mat def}
    Let $\mathcal{G} = (G,I,O,\lambda)$ be a labelled open graph with $\lambda(v) \in \{ X, Z \}$ for all $v \in \overline{O}$. Let $G_{disc}$ be $G$ with $Z$ labelled vertices disconnected from the rest of the graph. We define the \textbf{flow matrix} $M_{\mathcal{G}}$ as the sum of the reduced adjacency matrix $A_{G_{disc}}\mid_{\overline{I}}^{\overline{O}}$ and matrix that has $1$ at the intersections of $v$ row and $v$ column for all $v$ with $\lambda(v) = Z$, and $0$ otherwise.
\end{definition}

\begin{theorem}\label{flow mat invertibility}
    Let $\mathcal{G} = (G, I, O, \lambda)$ be a labelled open graph with $\lambda(v) \in \{ X, Z \}$ for all $v \in \overline{O}$. Then $\mathcal{G}$ has focussed Pauli flow if and only if $M_{\mathcal{G}}$ is right-invertible over $\mathbb{F}_2$.
\end{theorem}
\begin{proof}
    Let $G_{disc}$ be as in the definition \ref{flow mat def}. Let $G_X = (V_X, E_X)$ be $G$ (equivalently $G_{disc}$) with $Z$-labelled vertices removed from the rest of the graph. Then, the following facts are equivalent:\begin{enumerate}[noitemsep]
        \item $\mathcal{G}$ has focussed Pauli flow,
        \item $\mathcal{G}_X := (G_X, I, O, \lambda_X)$ has focussed Pauli flow, where $\lambda_X = \lambda\mid_{\left\{v \in \overline{O} \mid \lambda(v) = X\right\}}$,
        \item $A_{G_X}\mid_{V_X\setminus I}^{V_X\setminus O}$ is right-invertible,
        \item $M_{\mathcal{G}_X}$ is right-invertible,
        \item $M_{\mathcal{G}_{disc}}$ is right-invertible, where $\mathcal{G}_{disc} = (G_{disc}, I, O, \lambda)$,
        \item $M_{\mathcal{G}}$ is right-invertible.
    \end{enumerate}
    
    $(1 \Leftrightarrow 2):$ follows from lemmata \ref{Z removal} and \ref{Z intro} -- the $Z$ labelled vertices can be removed and introduced without affecting flow existence; the existences of Pauli flow and focussed Pauli flow are equivalent by theorem \ref{foc flow exists}.
    
    $(2 \Leftrightarrow 3):$ follows from corollary \ref{red adj matrix invertibility}, as in $G_X$ there are no $Z$ labelled vertices.

    $(3 \Leftrightarrow 4):$ the notions of reduced adjacency matrix and flow matrix agree in the case of only $X$ measurements, thus $A_{G_X}\mid_{V_X\setminus I}^{V_X\setminus O} = M_{\mathcal{G}_X}$.

    $(4 \Leftrightarrow 5):$ if there are no $Z$ labelled vertices in $G_{disc}$, then $G_{disc} = G_X$ and matrices from $4$ and $5$ are equal. Otherwise, let $v \in \overline{O}$ be any vertex with $\lambda(v) = Z$. In the flow matrix $M_{\mathcal{G}_{disc}}$, the $v$ row and the $v$ column are $0$ everywhere except for the intersection, as $v$ is disconnected from other vertices by assumption. The intersection of the $v$ row and the $v$ column contains $1$ by the construction of the flow matrix. Thus, $M_{\mathcal{G}_{disc}}$ is right-invertible if and only if its minor obtained by removing the $v$ row and the $v$ column is right-invertible. Performance of such removals for all $Z$ labelled vertices results in $M_{\mathcal{G}_X}$, and thus $M_{\mathcal{G}_{disc}}$ is right-invertible if and only if $M_{\mathcal{G}_X}$ is.

    $(5 \Leftrightarrow 6):$ by construction of flow matrix, $M_{\mathcal{G}_{disc}} = M_{\mathcal{G}}$ and the equivalence follows.

    Thus, $1 \Leftrightarrow 6$, as required.
\end{proof}

The above theorem will be essential in the next subsection. See figure \ref{fig:ex233 1z flow mat} for an example.

\begin{figure}
    \centering
    \begin{subfigure}[b]{0.3\textwidth}
        \centering
        $\begin{pNiceArray}{ccc|ccc}[first-row,last-row=6,first-col,last-col,margin]
              & C & D & E & F & G & H & \\
            A & 1 & 1 & 0 & 0 & 0 & 0 & A\\
            B & 1 & 1 & 1 & 0 & 0 & 0 & B\\
            \hline
            C & 0 & 0 & 0 & 0 & 1 & 1 & C\\
            D & 0 & 0 & 0 & 0 & 1 & 1 & D\\
            E & 0 & 0 & 0 & 1 & 0 & 0 & E\\
              & C & D & E & F & G & H & \\
        \end{pNiceArray}$
        \caption{The reduced adjacency matrix of the graph in \ref{fig:ex233 no flow}. This matrix is not right-invertible, so there is no Pauli flow.}
        \label{fig:ex233 no flow mat}
    \end{subfigure}
    \hfill
    \begin{subfigure}[b]{0.3\textwidth}
        \centering
        $\begin{pNiceArray}{ccc|ccc}[first-row,last-row=6,first-col,last-col,margin]
              & C & D & E & F & G & H & \\
            A & 1 & 0 & 0 & 0 & 0 & 0 & A\\
            B & 1 & 0 & 1 & 0 & 0 & 0 & B\\
            \hline
            C & 0 & 0 & 0 & 0 & 1 & 1 & C\\
            D & 0 & 1 & 0 & 0 & 0 & 0 & D\\
            E & 0 & 0 & 0 & 1 & 0 & 0 & E\\
              & C & D & E & F & G & H & \\
        \end{pNiceArray}$
        \caption{The flow matrix of the graph in \ref{fig:ex233 1z flow}. This matrix is right-invertible, so there is Pauli flow. It differs from the matrix in \ref{fig:ex233 no flow mat} in vertex $D$.}
        \label{fig:ex233 1z flow mat}
    \end{subfigure}
    \hfill
    \begin{subfigure}[b]{0.36\textwidth}
        \centering
        $\begin{pNiceArray}{ccc|ccc}[first-row,last-row=6,first-col,last-col,margin]
              & C & D & E & F & G & H & \\
            A & y_C & y_D & 0 & 0 & 0 & 0 & A\\
            B & y_C & y_D & y_E & 0 & 0 & 0 & B\\
            \hline
            C & z_C & 0 & 0 & 0 & x_C & x_C & C\\
            D & 0 & z_D & 0 & 0 & x_D & x_D & D\\
            E & 0 & 0 & z_E & x_E & 0 & 0 & E\\
              & C & D & E & F & G & H & \\
        \end{pNiceArray}$
        \caption{The variable flow matrix of the open graph in \ref{fig:ex233} (ignoring the labellings). $z_v$ is a shorthand for $(1+x_v)(1+y_v)$ where $v \in \{C,D,E\}$.}
        \label{fig:ex233 var flow mat}
    \end{subfigure}
    
    \caption{Examples of the matrix interpretations of flows for labelled open graphs from \ref{fig:ex233}}
    \label{fig:ex233 mat}
\end{figure}

\subsection{Reduction to \texorpdfstring{$\MaxRank$}{MaxRank}}

Finally, we can transform $\FlowSearch$ into a special case of $\MaxRank$ problem.

\begin{definition}
    Let $(G,I,O)$ be any open graph, i.e.\ an input for $\FlowSearch$. We define the \textbf{variable flow matrix} $M'_{G,I,O}$ as a matrix obtained as follows. (As a reminder $B := V \setminus \{ I \cup O \}$.)\begin{enumerate}[noitemsep]
        \item Start with the reduced adjacency matrix $A_G\mid_{\overline{I}}^{\overline{O}}$.
        \item For each $v \in B$, multiply the $v$ row by a variable $x_v$ and the $v$ column by a variable $y_v$.
        \item For each $v \in B$, set the intersection of the $v$ row and the $v$ column to $(1+x_v)(1+y_v)$.
    \end{enumerate}
\end{definition}

For an example, consider figure \ref{fig:ex233 var flow mat}.

\begin{theorem}\label{main helper}
    The answer to the $(G,I,O)$ instance of $\FlowSearch$ is $True$ if and only if there exists a valuation to $\{0, 1 \}$ of all variables in $M'_{G,I,O}$ resulting in a right-invertible matrix, i.e.\ when the answer to the instance of $\MaxRank$ given by the pair $(M'_{G,I,O},\ |\overline{O} |)$ is ``YES''.
\end{theorem}
\begin{proof}
    By theorem \ref{red to x z}, we can consider only $\lambda$ with codomain $\{ X, Z \}$ and by theorem \ref{no i cap o}, we can can assume $I \cap O = \emptyset$.
    
    $(\Rightarrow)$: let $\lambda \colon \overline{O} \to \{ X, Z \}$ be a measurement labelling for which $(G,I,O,\lambda)$ has focussed Pauli flow. For $v \in B$, send $x_v$ and $y_v$ to $1$ when $\lambda(v) = X$ and to $0$ when $\lambda(v) = Z$. Then, under this valuation, $M'_{G,I,O}$ evaluates to $M_{(G,I,O,\lambda)}$ which is right-invertible by theorem \ref{flow mat invertibility}.

    $(\Leftarrow)$: let $\sigma$ be a valuation sending variables in $M'_{G,I,O}$ to $\{ 0, 1 \}$ for which $M'_{G,I,O}$ is right-invertible. Suppose $\sigma(x_v) = 0 \ne 1 = \sigma(y_v)$ for some $v \in B$. Under $\sigma$, the $v$ row in $M'_{G,I,O}$ equals $0$ everywhere, and the matrix cannot be right-invertible as it does not have maximal row rank, contradiction. Now, suppose $\sigma(x_v) = 1 \ne 0 = \sigma(y_v)$ for some $v \in B$. Under $\sigma$, the $v$ column in $M'_{G,I,O}$ equals $0$ everywhere. Changing $\sigma(y_v)$ to $1$ could change the entries of the $v$ column, otherwise leaving the matrix unchanged. Thus, such a change cannot lower the row rank and $M'_{G,I,O}$ would stay right-invertible. Hence, there exists a valuation $\sigma'$ resulting in right-invertible matrix with $\sigma'(x_v)=\sigma'(y_v)$ for all $v \in B$. Now, we define $\lambda \colon \overline{O} \to \{ X, Z \}$ as follows. For $i \in I$: $\lambda(i) = X$. For $v \in B$, set $\lambda(v) = X$ when $\sigma'(x_v) = \sigma'(y_v) = 1$ and $\lambda(v) = Z$ when $\sigma'(x_v) = \sigma'(y_v) = 0$. Then, $M_{(G,I,O,\lambda)}$ equals $M'_{G,I,O}$ under $\sigma'$. Thus, $M_{(G,I,O,\lambda)}$ is right-invertible and by theorem \ref{flow mat invertibility}, $(G,I,O,\lambda)$ has focussed Pauli flow, ending the proof.
\end{proof}

Two variables $x_v$ and $y_v$ for each $v \in B$ might seem unnecessary. Multiplying both the row and the column by $x_v$ and setting the intersection to $1+x_v$ would also work. The problem is that then each variable no longer appears in one column or row only, and theorem \ref{main coco result} would not work.

\begingroup
\def\thetheorem{\ref{main result}}
\begin{theorem}[Repeated]
    $\FlowSearch$ is in $\RM{RP}$.
\end{theorem}
\addtocounter{theorem}{-1}
\endgroup
\begin{proof}
    By theorem \ref{main helper}, answering $(G,I,O)$ instance of $\FlowSearch$ is equivalent to answering instance of $\MaxRank$ given by $(M'_{G,I,O}, |\overline{O}|)$ over the field $\mathbb{F}_2$. The entries of $M'_{G,I,O}$ are all multi-affine -- they are either $0$, $1$, a variable, a product of two different variables or $(x_v+1)(y_v+1)$ for some $v \in B$. Hence, such matrix instances satisfy all conditions of theorem \ref{main coco result}, and the problem is in $\RM{RP}$.
\end{proof}

An analogous procedure can also determine whether a partial labelling $\lambda \colon \overline{O} \hookrightarrow \{ X, Z \}$ can be extended to a full labelling with flow. To do so, rather than using $\MaxRank$ on $M'_{G,I,O}$, we can consider it on $M'_{G,I,O}$ with some variables evaluated to $1$ or $0$ depending on $\lambda$. It means, that the problem of finding $\lambda$ such that $(G,I,O,\lambda)$ has Pauli flow is also solvable in random polynomial time, as captured by the following corollary with initial $\lambda = \emptyset$.

\begin{corollary}
    Given an open graph $(G,I,O)$ and a partial measurement labelling $\lambda$ with codomain $\{ X , Z \}$, it is in $\RM{RP}$ to check if $\lambda$ can be extended to a full labelling $\lambda'$ such that $(G,I,O,\lambda')$ has Pauli flow.
\end{corollary}

The pseudocodes of the algorithms arising from the above theorems and theorem \ref{main coco result} can be found in the appendix \ref{algo pseudocode}. In the appendix \ref{algo complexity}, we discuss the complexity and the possible implementation.

\subsection{Inputs and outputs}
Here, we show that given a labelled open graph with Pauli flow, it is always possible to reduce the number of outputs to match the number of inputs while preserving the existence of Pauli flow for some measurement labelling.

\begingroup\def\thetheorem{\ref{removing outputs}}\begin{theorem}[Repeated] Suppose $(G,I,O,\lambda)$ has Pauli flow and $|O|>|I|$. Then there exists a subset $O' \subseteq O$ such that $|O'| = |I|$ and a labelling $\lambda'$ such that $(G,I,O',\lambda')$ has Pauli flow.
\end{theorem}\addtocounter{theorem}{-1}\endgroup
\begin{proof}
    Let $M = M_{(G,I,O,\lambda)}$ be the flow matrix of $(G,I,O,\lambda)$ with Pauli flow and with $|O|>|I|$. By theorem \ref{flow mat invertibility}, $M$ is right-invertible. Hence, $M$ has $|\overline{O}| \times |\overline{O}|$ invertible minor. Let $C$ be the set of vertices corresponding to the columns that are not in such minor. Suppose that $o \in C \cap O$. The output $o$ can be removed from the graph without breaking the flow existence -- $M$ without $o$ column still contains an invertible square minor of maximal size. By lemma \ref{Z intro}, we can equivalently change $o$ to be $Z$ labelled without breaking flow existence. Now, assume that $C \cap O = \emptyset$. In the flow matrix, the column and row of $v$ with $\lambda(v) = Z$ are $0$ except for the intersection. Therefore, the $v$ column must be included in the maximal invertible minor and thus $v \notin C$. Hence, $C$ is a subset of the set of $X$ labelled vertices. Let $v \in C$. We can remove $v$ from the graph, keeping the flow existence: removal of $v$ column does not impact right-invertibility, and removal of $v$ row cannot break right-invertibility either -- other rows would remain linearly independent. By lemma \ref{Z intro}, $v$ can be reintroduced with the same neighbours as previously, but with $\lambda(v) = Z$. The same change can be applied to all vertices in $C$. Thus, if $|O|>|I|$ it is always possible to decrease the number of outputs or $X$ measured vertices. As the second does not affect $|O|>|I|$, by repeating this process, eventually the number of outputs must decrease and match $|I|$.
\end{proof}

For examples, see figure \ref{fig:ex233 outputs} in the appendix \ref{reducing number of outputs example}. The procedure used in the proof above can be efficiently implemented by utilizing the basis-finding algorithm i.e.\ Gaussian elimination. The basis-finding procedure has another usage. In the case of $X$ labelling only, suppose that we are given $O$ but not $I$. We can then find $I$ with $|I|=|O|$ resulting in flow, or determine that no flow exists for any set of inputs.

\begin{lemma}
    Let $G$ be a graph and $O \subseteq V$. Let $\lambda \colon \overline{O} \to \{X \}$. Then either $(G,I,O,\lambda)$ does not have Pauli flow for any $I$, or $(G,I,O,\lambda)$ has Pauli flow for some $I$ with $|I| = |O|$.
\end{lemma}
\begin{proof}
    Consider $(G,\emptyset,O,\lambda)$. Suppose that it does not have Pauli flow. Then, the reduced adjacency matrix of $(G,\emptyset,O)$ is not right-invertible. Changing inputs to a different set than $\emptyset$ corresponds to the removal of columns but not rows from the reduced adjacency matrix -- it cannot make the matrix right-invertible. Hence there is no $I$ for which $(G,I,O,\lambda)$ has Pauli flow. Conversely, if $(G,\emptyset,O,\lambda)$ has Pauli flow, then we can choose $|\overline{O}| \times |\overline{O}|$ minor from the reduced adjacency matrix. The columns that are not chosen can be removed without breaking the flow, i.e.\ the corresponding vertices can be changed to be inputs. Note, that some output could be changed to also be an input. This process always turns $|V|-|\overline{O}| = |O|$ vertices into inputs, which ends the proof.
\end{proof}

Finally, again in the case of $X$ labelling only, suppose that we are given inputs. Can we find a minimal (smallest) set of outputs resulting in Pauli flow? The answer is yes.

\begin{lemma}
    Let $G$ be a graph and $I \subseteq V$. Let $\lambda(v) = X$ for $v \in V$. Then a minimal $O$ resulting in $(G,I,O,\lambda\mid_{\overline{O}})$ having Pauli flow can be efficiently found.
\end{lemma}
\begin{proof}
    Let $M_{\emptyset}$ be the reduced adjacency matrix of $(G,I,\emptyset,\lambda)$. Let $D$ be a set of rows forming the basis of the space given by all rows. Let $C$ be the set of vertices whose rows are not in $D$. Then $O = C$ is the required minimal set -- clearly $(G,I,O,\lambda\mid_{\overline{O}})$ has Pauli flow -- in its reduced adjacency matrix $M_O$ the rows are linearly independent, so the matrix is right-invertible. Also, $O$ is minimal -- any smaller set $O'$ of outputs cannot result in a right-invertible reduced adjacency matrix, as such matrix would have more rows than $M_O$, and all of its rows would be from $M_{\emptyset}$. But $M_O$ has the maximal number of linearly independent rows, as those rows form a basis of the space spanned by the rows of $M_{\emptyset}$. Note, that some inputs could be changed to also be outputs.
\end{proof}

\section{Conclusions and further work}\label{conclusions and further work}
We have shown that given an open graph it is in $\RM{RP}$ to determine whether there exists a measurement labelling for which the open graph has Pauli flow. In other words, there is a random polynomial time algorithm deciding whether an open graph state can be used for any type of deterministic computation. To obtain this result, we developed an algebraic interpretation of flow for the case with two Pauli measurements $X$ and $Z$, and then performed a reduction to a known problem from computational complexity. We have also shown that the algebraic interpretation can be useful when looking for the necessary properties that a set of inputs or outputs must satisfy. In particular, we showed that it is always possible to reduce the number of outputs to match the number of inputs when it is allowed to change some measurement labels to $Z$. Our results contribute to the general picture of what the Pauli flow structure is and which open graphs can exhibit it.

Sometimes a graph state can be prepared, but it might be difficult to check if such a state can be useful in MBQC. A possible approach can be checking other states in orbit \cite{adcockMappingGraphState2020}. Our result can be interpreted as answering whether an open graph can be used in any form of quantum computation, or to give some conditions on such computation that are necessary to achieve deterministic labelled open graphs.

The remainder is a discussion of possible future work.

\paragraph{Polynomial time}
When showing that some problem, $\FlowSearch$ in our case, is in $\RM{RP}$, it is natural to ask whether a problem is also in $\RM{P}$. It would be interesting to modify reduction to $\MaxRank$ to obtain instances for which polynomial time algorithms exist, like the instances in \cite{ivanyosDeterministicPolynomialTime2010}.

\paragraph{Other measurements}
The main trick we have used to show $\FlowSearch \in \RM{RP}$, was the restriction of possible measurements to Pauli $X$ and $Z$ measurements for which we developed algebraic interpretation. A similar interpretation of $XY$ planar measurements is known \cite{mhallaWhichGraphStates2014a} and may be extended to allow Pauli $Z$ measurements, but the structure is harder to work with due to the partial order requirements. Extending our results to also work for $XY$ and $Z$ measurements is an interesting direction for further research -- $X$ and $Z$ measurements result in the Clifford fragment which can be classically simulated \cite{gottesmanHeisenbergRepresentationQuantum1998a}, while $XY$ planar measurements are sufficient for universality even on cluster states \cite{mantriUniversalityQuantumComputation2017}. It would also be interesting to check if our approach can be used to minimise the number of vertices that need to be $Z$ measured to have flow. The main problem with adapting the presented methods to finding labels with planar measurements is the order. Without verifying the order (ignoring conditions \ref{P1}, \ref{P2}, \ref{P3}), a reduction to the $\MaxRank$ problem should still be possible. However, order changes everything. There might be a flow where a particular vertex has a label $X$ or $Z$. Yet, measuring such vertex in any plane can break the flow.

\paragraph{Altering the set of edges}
Another problem that can be interpreted as $\MaxRank$ instance is finding how to change the set of edges to get an open graph with flow. For instance, when all non-outputs are $X$ labelled, we can put a new variable in the place of $0$ entries corresponding to the lack of an edge between vertices. An interesting question would be to determine how many edges must be added to get flow or how many must be removed. If such a problem is tractable, it could have practical usage, where given an open graph state one could say how far the state is from one allowing deterministic computation. The number of edges that need to be flipped corresponds to the number of $CZ$ gates that must be applied.

\paragraph{Circuit extraction}
Our results also connect to the ZX calculus, where Pauli flow is a necessary condition for efficient circuit extraction \cite{simmonsRelatingMeasurementPatterns2021}. It would be interesting to expand on this connection, for instance, by attempting the classification of small open graphs, that, when contained in ZX diagram, are guaranteed to break flow, and thus should be avoided in any form of optimization utilizing ZX. Another connection could be circuit extraction from phase-free ZH -- $X$ and $Z$ measurements are sufficient for universality in MBQC on hypergraph states \cite{takeuchiQuantumComputationalUniversality2019}. There, the errors can be multi-qubit byproducts and their fixing is done by switching measurement bases. Our method of checking the existence of measurement labelling resulting in Pauli flow could be helpful when looking for measurements that can be swapped without removing the property of determinism (i.e.\ Pauli flow) in the subparts without any hyperedges.

\section*{Acknowledgement}
I thank Miriam Backens, Tommy McElvanney, and Korbinian Staudacher for helpful discussions. Special thanks go to my supervisor Miriam Backens, for sharing theorem \ref{red adj matrix invertibility with DAG} with me. I also thank anonymous reviewers for their useful comments.

\phantomsection

\addcontentsline{toc}{chapter}{Bibliography}

% Bibliography
\begin{sloppypar}
\bibliographystyle{eptcs}
\bibliography{myrefdb}
\end{sloppypar}

\appendix
\counterwithin{figure}{section}
\counterwithin{table}{section}

\section{More about \texorpdfstring{$\MaxRank$}{MaxRank}}\label{more about maxrank}

In order to prove theorem \ref{main coco result}, we need the following lemma.

\begin{lemma}\label{0 small 0 big}
    A multi-affine polynomial is $0$ over a finite field $\mathbb{F}_s$ if and only if it is $0$ over $\mathbb{F}_{s^k}$ for any $k \in \mathbb{Z}_+$.
\end{lemma}
\begin{proof}
    This follows from \cite[Lemma 25 and Corollary 26]{bussComputationalComplexityProblems1999}.
\end{proof}

Because of the above, instead of testing whether a multi-affine polynomial is $0$ over $\mathbb{F}_s$, we can test whether it is $0$ over some large field extension. In particular, the extension can be taken sufficiently large to ensure that the Schwartz-Zippel lemma \cite{ore1921höhere, schwartzFastProbabilisticAlgorithms1980, demilloProbabilisticRemarkAlgebraic1978, zippelProbabilisticAlgorithmsSparse1979} applies (adapted to $\mathbb{F}_{s^k}$ only):

\begin{theorem}\label{schwartz-zippel}
    Let $P \in \mathbb{F}_{s^k}[x_1, \dots, x_n]$ be a non-zero polynomial of total degree $d$ over $\mathbb{F}_{s^k}$. Let $a_1, \dots, a_n \in \mathbb{F}_{s^k}$ be chosen at random uniformly. Then:
    \begin{gather*}
        \Pr[P(a_1,\dots,a_n) = 0] \le \frac{d}{s^k}.
    \end{gather*}
\end{theorem}

We can now prove theorem \ref{main coco result}.

\begin{proof}[Proof of theorem \ref{main coco result}]
    Let $M, r$ be a matrix and an integer forming an input to $\MaxRank$ satisfying the conditions from the theorem statement. If $r > \min (m,n)$ or $r < 0$, the answer is ``NO'' and can be returned immediately, so assume $0 \le r \le \min (m,n)$. Consider any $r \times r$ minor $M'$ of $M$. We show $\det M'$ is multi-affine: consider any variable $x$ in $M'$. Then, $x$ appears in at most one row or one column of $M'$. By performing Laplace expansion on $M'$ in such row (column), we find that $\det M'$ is a sum of determinants of $(r-1) \times (r-1)$ minors of $M'$ that do not contain $x$ multiplied by elements of the row (column) used for the expansion that itself may contain $x$ only in degree $1$ due to multi-affinity assumption about matrix entries. Therefore, $x$ appears in degree at most $1$ in $\det M'$ and the same for all other variables of $\det M'$, so the determinant is multi-affine. Therefore, the method from \cite[Theorem 28]{bussComputationalComplexityProblems1999} applies. We present a modified version for clarity.

    Consider the following procedure. Let $p \le \frac{1}{2}$ be the desired error probability. It is sufficient to consider $p = \frac{1}{2}$, but we present also how to achieve arbitrarily small error probability. Given $M, r, p$, let $k$ be such that $\mathbb{F}_{s^k}$ has at least $\frac{t}{p}$ elements, i.e.\ $k = \left\lceil \log_s \frac{t}{p} \right\rceil$. Let $a_1, \dots, a_t$ be a randomly chosen valuation of $x_1, \dots, x_t$ from $\mathbb{F}_{s^k}$. Let $r_a = \rank M(a_1, \dots, a_t)$, which can be found by Gaussian elimination i.e.\ in polynomial time. The computations over finite field are possible in time polynomial in $\log_2 s$ and $k$ \cite{gashkovComplexityComputationFinite2013}. Return ``YES'' if and only if $r_a \ge r$. We show, that the following procedure shows $\RM{RP}$ containment.

    Suppose, that the actual answer to the instance is ``YES''. Then, under some valuation, $M$ has rank at least $r$. Under such valuation, $M$ must have $r \times r$ reversible minor. Let $M'$ be such minor. By the previous part, $\det M'$ is multi-affine. Let $d$ be the total degree of $\det M'$. Then, $d \le t$ again by multi-affinity. By lemma \ref{0 small 0 big}, $M'$ is $0$ over $\mathbb{F}_s$ if and only if it is $0$ over $\mathbb{F}_{s^k}$. Combining everything with the Schwartz-Zippel theorem \ref{schwartz-zippel}, we get that:
    \begin{gather*}
        \Pr[\det M'(a_1,\dots,a_t) =_{\mathbb{F}_{s^k}} 0] \le \frac{d}{s^k} \le \frac{t}{s^k} \le \frac{t}{t / p} = p
    \end{gather*}
    The same holds for all $r \times r$ minors of $M$ that are invertible under some valuation. Hence, with error probability at most $p$, a random valuation from $\mathbb{F}_{s^k}$ results in a non-root of some $r \times r$ minor's determinant of $M$ in which case rank of the minor under such valuation is $r$ and so $\rank M$ is at least $r$, i.e.\ the procedure above would find $r_a \ge r$ and return ``YES''. Hence, the procedure described above returns the correct answer with probability at least $1-p \ge \frac{1}{2}$.

    Now suppose, that the answer to the instance is ``NO''. Then, all $r \times r$ minors of $M$ must have determinants equal $0$ over $\mathbb{F}_s$. By multi-affinity, they are also equal $0$ over $\mathbb{F}_{s^k}$. Hence, a random valuation $a_1, \dots, a_t$ always results in $\rank M(a_1, \dots, a_t) \le k$. Hence, the procedure described above returns ``NO'' with probability $1$, ending the proof of containment in $\RM{RP}$.
\end{proof}

\begin{example}
    Consider the following matrix over $\mathbb{F}_2$:
    \begin{gather*}
    M = \begin{pNiceArray}{cccc}
        x_1 & 0 & 1 & 0\\
        0 & x_2 x_3 & 0 & 0 \\
        1 & 0 & x_1 & 0 \\
        x_1 & x_2 & x_3 & 0
    \end{pNiceArray}
    \end{gather*}
    Setting $x_1 = 0$ and $x_2, x_3 = 1$ results in $M$ having rank $3$. However, there is no valuation resulting in the matrix having rank $4$, as the fourth column contains $0$s only. Hence, the answer to instances $(M,1), (M,2), (M,3)$ of $\MaxRank$ is ``YES'', but the answer to $(M,4)$ is ``NO''.
\end{example}

Some closely related problems also defined in \cite{bussComputationalComplexityProblems1999} include $\problem{MinRank}$, $\problem{Sing}$, and $\NonSing$. $\problem{MinRank}$ takes the same inputs as $\MaxRank$ and asks whether a rank $\le r$ can be achieved. $\problem{Sing}$ and $\problem{NonSing}$ take a square matrix and ask whether the matrix can be made singular and non-singular respectively.

In general, these problems are hard or sometimes unsolvable. For instance, $\problem{MinRank}$ is undecidable when $R = \mathbb{Z}, E=S=\{0,1\}$ \cite{bussComputationalComplexityProblems1999}. The problems are very natural and often appear when working on any linear algebra problems. $\problem{Sing}$ is useful in cryptography due to its hardness (for example, see \cite{bardetImprovementsAlgebraicAttacks2020}). It is also interesting from a complexity perspective (for example, see \cite{mahajanComplexityMatrixRank2010}).

We only work with $\MaxRank$ over finite fields (in fact, later we only consider $\mathbb{F}_{2^k}$).  A variant where each variable can appear at most once is in $\RM{P}$. A slightly more general version where a variable can appear in at most one row or one column but an unlimited number of times is also in $\RM{P}$ \cite{ivanyosDeterministicPolynomialTime2010}. Note, that this result is not stronger than the presented theorem \ref{main coco result}, as the entries of the matrix there cannot include products of variables. When each variable can appear at most twice in the matrix, but not necessarily in one row or one column, the problem already becomes $\RM{NP}$-complete \cite{harveyComplexityMatrixCompletion2006}.

\section{Reduced adjacency matrix invertibility for \texorpdfstring{$X$}{X} and \texorpdfstring{$XY$}{XY} measurements}\label{red adj matrix invertibility with DAG proof}

\begin{figure}[H]
    \centering
    \begin{subfigure}[b]{0.4\textwidth}
        \centering
        $\scalebox{1}{\begin{tikzpicture}
	\begin{pgfonlayer}{nodelayer}
		\node [style=Small Empty Box] (0) at (-2, 2.5) {};
		\node [style=Small Empty Box] (1) at (-2, -2.5) {};
		\node [style=Small Empty] (5) at (2, 2.5) {};
		\node [style=Small Black] (6) at (2, 0) {};
		\node [style=Small Empty] (7) at (2, -2.5) {};
		\node [style=none] (8) at (0, 4) {};
		\node [style=none] (9) at (0, -3.5) {};
		\node [style=none] (10) at (4.5, 0) {};
		\node [style=none] (11) at (-4.5, 0) {};
		\node [style=none] (12) at (-2, 3.25) {$R$ in $XY$};
		\node [style=none] (13) at (-2, -3.25) {$S$ in $XY$};
		\node [style=none] (17) at (2, 3.25) {$V$};
		\node [style=none] (18) at (2, 0.75) {$T$ in $XY$};
		\node [style=none] (19) at (2, -1.75) {$W$};
		\node [style=Small Black] (22) at (-2, 2.5) {};
		\node [style=Small Black] (23) at (-2, -2.5) {};
		\node [style=Small Black] (24) at (-2, 0) {};
		\node [style=none] (25) at (-3.25, 0.75) {$U$ in $X$};
	\end{pgfonlayer}
	\begin{pgfonlayer}{edgelayer}
		\draw (22) to (24);
		\draw (24) to (23);
		\draw (22) to (5);
		\draw (24) to (6);
		\draw (7) to (23);
	\end{pgfonlayer}
\end{tikzpicture}}$
        \caption{A labelled open graph with two inputs and two outputs containing Pauli flow. Vertex $U$ is $X$ labelled and all other non-outputs are $XY$ labelled.}
        \label{fig:ex x xy d}
    \end{subfigure}
    \hfill
    \begin{subfigure}[b]{0.4\textwidth}
        \centering
        $\scalebox{1}{\begin{tikzpicture}
	\begin{pgfonlayer}{nodelayer}
		\node [style=Tiny Black] (2) at (2, 2.5) {};
		\node [style=Tiny Black] (3) at (2, 0) {};
		\node [style=Tiny Black] (4) at (2, -2.5) {};
		\node [style=none] (5) at (0, 4) {};
		\node [style=none] (6) at (0, -3.5) {};
		\node [style=none] (7) at (4.5, 0) {};
		\node [style=none] (8) at (-4.5, 0) {};
		\node [style=none] (9) at (-2, 3.25) {$R$};
		\node [style=none] (10) at (-2, -1.75) {$S$};
		\node [style=none] (11) at (2, 3.25) {$V$};
		\node [style=none] (12) at (1.5, 0.75) {$T$};
		\node [style=none] (13) at (1.5, -1.75) {$W$};
		\node [style=Tiny Black] (14) at (-2, 2.5) {};
		\node [style=Tiny Black] (15) at (-2, -2.5) {};
		\node [style=Tiny Black] (16) at (-2, 0) {};
		\node [style=none] (17) at (-2, 0.75) {$U$};
	\end{pgfonlayer}
	\begin{pgfonlayer}{edgelayer}
		\draw [style=directed] (2) to (14);
		\draw [style=directed, bend right=330] (3) to (16);
		\draw [style=directed, bend left] (16) to (3);
		\draw [style=directed] (4) to (15);
		\draw [style=directed] (2) to (3);
		\draw [style=directed] (4) to (3);
	\end{pgfonlayer}
\end{tikzpicture}}$
        \caption{A directed graph $F$ for the labelled open graph from \ref{fig:ex x xy d}, obtained from theorem \ref{red adj matrix invertibility with DAG}. It contains a cycle. The corresponding $F'$ contains only one edge $TU$ and is acyclic.}
        \label{fig:ex x xy d2}
    \end{subfigure}
%\iffalse
    \caption{An example of a labelled open graph with both $X$ and $XY$ measurements and explanation of theorem \ref{red adj matrix invertibility with DAG} acting on it.}
\end{figure}
\begin{figure}[H]
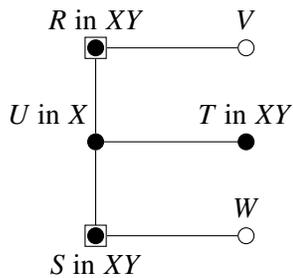
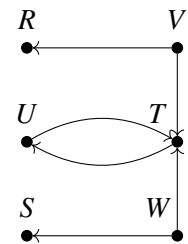
\ContinuedFloat
%\fi
    \begin{subfigure}[b]{0.4\textwidth}
        \centering\begin{gather*}
        \begin{pNiceArray}{cc|cc}[first-row,last-row,first-col,last-col,margin]
              & T & U & V & W & \\
            R & 0 & 1 & 1 & 0 & R\\
            S & 0 & 1 & 0 & 1 & S\\
            \hline
            T & 0 & 1 & 0 & 0 & T\\
            U & 1 & 0 & 0 & 0 & U\\
              & T & U & V & W & \\
        \end{pNiceArray}
        \end{gather*}
        \caption{Reduced adjacency matrix of the labelled open graph from \ref{fig:ex x xy d}.}
        \label{fig:ex x xy mat}
    \end{subfigure}\hfill
    \begin{subfigure}[b]{0.4\textwidth}
        \centering\begin{gather*}
        \begin{pNiceArray}{cc|cc}[first-row,last-row,first-col,last-col,margin]
              & R & S & T & U & \\
            T & 0 & 0 & 0 & 1 & T\\
            U & 0 & 0 & 1 & 0 & U\\
            \hline
            V & 1 & 0 & 1 & 0 & V\\
            W & 0 & 1 & 1 & 0 & W\\
              & R & S & T & U & \\
        \end{pNiceArray}
        \end{gather*}
        \caption{The inverse of matrix in \ref{fig:ex x xy mat}, i.e. the matrix $A_F\mid_{\overline{O}}^{\overline{I}}$ where $F$ is as in \ref{fig:ex x xy d2}.}
    \end{subfigure}
    \captionsetup{list=off,format=cont}
    \caption{An example of a labelled open graph with both $X$ and $XY$ measurements and explanation of theorem \ref{red adj matrix invertibility with DAG} acting on it.}
    \label{fig:ex x xy full}
\end{figure}

\section{Pseudocode for algorithms}\label{algo pseudocode}
In the pseudocode below, we do not explicitly construct an instance of matrix used for $\MaxRank$ problem, as, in practice, it might be difficult to explicitly construct a matrix with variables that can be substituted with values from large field extensions of $\mathbb{F}_2$. Instead, we only construct the matrix under some valuation.

\algrenewcommand\algorithmicdo{}
\algrenewcommand\algorithmicthen{}
\begin{algorithmic}[1]
\Statex Checks if a partial $X, Z$ labelling can be extended so that $(G,I,O,\lambda)$ has Pauli flow. The error probability must be bounded above by $p$.
\Procedure{FlowSearchAux}{$G,I,O,\lambda,p$}
    \State $M \leftarrow A_{G}\mid_{\overline{I}}^{\overline{O}}$ \Comment{Construction of reduced adjacency matrix}
    \State $Vars \leftarrow \emptyset$ \Comment{Initialize set of variables}
    \For {$v \in B$} \Comment{Detecting unlabelled vertices}
        \If {$\lambda(v)$ is defined}
            \If {$\lambda(v) == Z$} \Comment{\parbox[t]{.5\linewidth}{Updating the row and the column of $Z$ labelled vertex}}
                \State multiply $v$ row of $M$ by $0$
                \State multiply $v$ column of $M$ by $0$
                \State set the intersection of $v$ row and $v$ column of $M$ to $1$
            \EndIf
        \Else
            \State $Vars \leftarrow Vars \cup \{ x_v, y_v \}$
        \EndIf
    \EndFor
    \State $k \leftarrow \left\lceil \log_2 \frac{|Vars|}{p} \right\rceil$ \Comment{\parbox[t]{.5\linewidth}{Minimal $k$ such that $\mathbb{F}_{2^k}$ has at least $\frac{|Vars|}{p}$ elements and the error probability is below $p$.}}
    \State Randomly sample $\sigma \colon Vars \to \mathbb{F}_{2^k}$
    \For {$v \in B$} \Comment{\parbox[t]{.5\linewidth}{Construction of $M'_{G,I,O}$ under valuation $\sigma$, computations are done in $\mathbb{F}_{2^k}$}}
        \State multiply $v$ row of $M$ by $\sigma(x_v)$
        \State multiply $v$ column of $M$ by $\sigma(y_v)$
        \State set the intersection of $v$ row and $v$ column to $(\sigma(x_v)+1)(\sigma(y_v)+1)$
    \EndFor
    \State Gaussian eliminate $M$
    \State \Return $(\rowrank M == |\overline{O}|)$
\EndProcedure
\[\]
\Statex Main algorithm, returns $True$ if $(G,I,O,\lambda)$ has Pauli flow for some $\lambda$. The error probability must be bounded above by $p$.
\Procedure{FlowSearch}{$G,I,O,p$}
    \State \Return{\Call{FlowSearchAux}{$G,I,O,\emptyset,p$}}
\EndProcedure
\[\]
\Statex Finds $\lambda$ such that $(G,I,O,\lambda)$ has Pauli flow. Initially, checks whether any such $\lambda$ exists, up to error probability $p$.
\Procedure{FindLabelling}{$G,I,O,p$}
    \If {$ {\bf{not}}\ \Call{FlowSearch}{G,I,O,p}$}
        \State \Return{``NO $\lambda$ EXISTS''}
    \EndIf
    \State $\lambda \leftarrow \emptyset$ \Comment{Initialization of $\lambda$}
    \For{$v \in I$}
        \State $\lambda(v) \leftarrow X$ \Comment{Inputs must be $X$ labelled}
    \EndFor
    \For{$v \in B$}
        \State $confirm_v \leftarrow False$
        \State $current_v \leftarrow X$ \Comment{Initially attempt $X$ label}
        \While{${\bf{not}}\ confirm_v$} \Comment{Alternate $X$ and $Z$ labels until one works}
            \State $\lambda(v) \leftarrow current_v$
            \If {\Call{FlowSearchAux}{$G,I,O,\lambda,p$}} \Comment{\parbox[t]{.5\linewidth}{Note, that the error probability could be increased at the cost of possibly more tries being required.}}
                \State $confirm_v \leftarrow True$
            \Else
                \If {$current_v == X$}
                    \State $current_v \leftarrow Z$
                \Else
                    \State $current_v \leftarrow X$
                \EndIf
            \EndIf
        \EndWhile
    \EndFor
    \State \Return{$\lambda$}
\EndProcedure

\end{algorithmic}

\section{Complexity of algorithms}\label{algo complexity}
The most memory-expensive part of the algorithms is the creation of $M'_{G,I,O}$ under some valuation from $\mathbb{F}_{2^k}$ where $k$ depends on the desired error probability $p$ and the size of the input graph. Since $k = \left\lceil \log_2 \frac{|Vars|}{p} \right\rceil$ and $|Vars| \le 2 \cdot |B|$, we get that $k \in O\left(\log_2 \frac{|B|}{p}\right)$. The elements of such fields can be represented using $O(k)$ long vectors over $\mathbb{F}_2$ with the time complexity of basic arithmetic operations on such field bounded above by $O(k^2)$ \cite{gashkovComplexityComputationFinite2013}. Therefore, the memory requirement can be bounded above by $O\left(|\overline{O}|\times|\overline{I}|\times\log_2 \frac{|B|}{p}\right) \in O\left(n^2 \log_2 \frac{n}{p}\right)$ where $n = |V|$. The most time-expensive part of the algorithms are the Gaussian eliminations. Each Gaussian elimination requires $O\left(n^3\right)$ basic operations in $\mathbb{F}_{2^k}$. Thus, a (not very efficient) upper bound for the time complexity is $O\left(n^3 \log_2^2 \frac{n}{p}\right)$ for the decision variant and expected $O\left(n^4 \log_2^2 \frac{n}{p}\right)$ for the actual finding of the labelling resulting in a Pauli flow. Many programming languages offer packages for efficient computation in finite fields. For instance, in Python one can use Galois package \cite{hostetterGaloisPerformantNumPy2020fixed} which works very well for $\mathbb{F}_{2^k}$ with $k$ such that the precomputed Conway polynomial \cite{frankluebeckhomepagefixed} is known, for example, all $1 \le k \le 91$. Such values of $k$ are sufficient for all reasonable computations, as the error can be dropped below $\frac{1}{2^{50}}$. At that point, it is more likely for a random cosmic beam to corrupt the computation than to get an error due to the probabilistic nature of the algorithms.

\section{Figure for reducing the number of outputs}\label{reducing number of outputs example}
\begin{figure}[H]
    \centering
    \begin{subfigure}[b]{0.4\textwidth}
        \centering
        $\scalebox{1}{\begin{tikzpicture}
	\begin{pgfonlayer}{nodelayer}
		\node [style=Small Empty Box] (0) at (-4, 1.5) {};
		\node [style=Small Empty Box] (1) at (-4, -1.5) {};
		\node [style=Small Black] (2) at (0, 3) {};
		\node [style=Small Grey] (3) at (0, 0) {};
		\node [style=Small Black] (4) at (0, -3) {};
		\node [style=Small Empty] (5) at (4, 3) {};
		\node [style=Small Grey] (6) at (4, 0) {};
		\node [style=Small Empty] (7) at (4, -3) {};
		\node [style=none] (8) at (0, 5) {};
		\node [style=none] (9) at (0, -4) {};
		\node [style=none] (10) at (6, 0) {};
		\node [style=none] (11) at (-6, 0) {};
		\node [style=none] (12) at (-4, 2.25) {$A$};
		\node [style=none] (13) at (-4, -0.75) {$B$};
		\node [style=none] (14) at (0, 3.75) {$C$};
		\node [style=GreyTEXT] (15) at (0, 0.75) {$D$};
		\node [style=none] (16) at (0, -2.25) {$E$};
		\node [style=none] (17) at (4, 3.75) {$F$};
		\node [style=GreyTEXT] (18) at (4, 0.75) {$G$};
		\node [style=none] (19) at (4, -2.25) {$H$};
		\node [style=Small Black] (20) at (0, -3) {};
		\node [style=Small Black] (21) at (0, 3) {};
		\node [style=Small Black] (22) at (-4, 1.5) {};
		\node [style=Small Black] (23) at (-4, -1.5) {};
	\end{pgfonlayer}
	\begin{pgfonlayer}{edgelayer}
		\draw (0) to (2);
		\draw [style=grey] (0) to (3);
		\draw (2) to (1);
		\draw [style=grey] (3) to (1);
		\draw (1) to (4);
		\draw (4) to (5);
		\draw [style=grey] (2) to (6);
		\draw [style=grey] (3) to (6);
		\draw [style=grey] (3) to (7);
		\draw (7) to (2);
	\end{pgfonlayer}
\end{tikzpicture}}$
        \caption{The open graph from \ref{fig:ex233 1z flow} with flow matrix from \ref{fig:ex233 1z flow mat}. Columns given by vertices $C$, $D$, $E$, $F$ and $H$ form the basis. Thus, the output $G$ can be removed (equivalently: changed to be $Z$ labelled).}
        \label{fig:ex233 2o flow}
    \end{subfigure}
    \hfill
    \begin{subfigure}[b]{0.4\textwidth}
        \centering
        $\scalebox{1}{\begin{tikzpicture}
	\begin{pgfonlayer}{nodelayer}
		\node [style=Small Empty Box] (24) at (-4, 1.5) {};
		\node [style=Small Empty Box] (25) at (-4, -1.5) {};
		\node [style=Small Black] (26) at (0, 3) {};
		\node [style=Small Black] (27) at (0, 0) {};
		\node [style=Small Black] (28) at (0, -3) {};
		\node [style=Small Empty] (29) at (4, 3) {};
		\node [style=Small Empty] (30) at (4, 0) {};
		\node [style=Small Empty] (31) at (4, -3) {};
		\node [style=none] (32) at (0, 5) {};
		\node [style=none] (33) at (0, -4) {};
		\node [style=none] (35) at (-6, 0) {};
		\node [style=none] (36) at (-4, 2.25) {$A$};
		\node [style=none] (37) at (-4, -0.75) {$B$};
		\node [style=none] (38) at (0, 3.75) {$C$};
		\node [style=none] (39) at (0, 0.75) {$D$};
		\node [style=none] (40) at (0, -2.25) {$E$};
		\node [style=none] (41) at (4, 3.75) {$F$};
		\node [style=none] (42) at (4, 0.75) {$G$};
		\node [style=none] (43) at (4, -2.25) {$H$};
		\node [style=Small Black] (44) at (0, -3) {};
		\node [style=Small Black] (45) at (0, 3) {};
		\node [style=Small Black] (46) at (-4, 1.5) {};
		\node [style=Small Black] (47) at (-4, -1.5) {};
		\node [style=none] (48) at (6, 0) {};
	\end{pgfonlayer}
	\begin{pgfonlayer}{edgelayer}
		\draw (24) to (26);
		\draw (24) to (27);
		\draw (26) to (25);
		\draw (27) to (25);
		\draw (25) to (28);
		\draw (28) to (29);
		\draw (26) to (30);
		\draw (27) to (30);
		\draw (31) to (26);
	\end{pgfonlayer}
\end{tikzpicture}}$
        \caption{An example of a labelled open graph with Pauli flow in which no output can be immediately changed to be $Z$ measured. However, changing $D$ to be $Z$ measured makes it possible to also change output $G$ to be $Z$ measured.}
        \label{fig:ex233 changed}
    \end{subfigure}
    \begin{subfigure}[b]{0.4\textwidth}
        \centering\begin{gather*}
        \begin{pNiceArray}{ccc|ccc}[first-row,last-row=6,first-col,last-col,margin]
              & C & D & E & F & G & H & \\
            A & 1 & 1 & 0 & 0 & 0 & 0 & A\\
            B & 1 & 1 & 1 & 0 & 0 & 0 & B\\
            \hline
            C & 0 & 0 & 0 & 0 & 1 & 1 & C\\
            D & 0 & 0 & 0 & 0 & 1 & 0 & D\\
            E & 0 & 0 & 0 & 1 & 0 & 0 & E\\
              & C & D & E & F & G & H & \\
        \end{pNiceArray}
        \end{gather*}
        \caption{Flow matrix of the open graph in \ref{fig:ex233 changed}. Columns given by vertices $C$, $E$, $F$, $G$, $H$ form the basis. Thus, the vertex $D$ can be changed to be $Z$ labelled.}
        \label{fig:ex233 changed matrix}
    \end{subfigure}\hfill
    \begin{subfigure}[b]{0.4\textwidth}
        \centering\begin{gather*}
        \begin{pNiceArray}{ccc|ccc}[first-row,last-row=6,first-col,last-col,margin]
              & C & D & E & F & G & H & \\
            A & 1 & 0 & 0 & 0 & 0 & 0 & A\\
            B & 1 & 0 & 1 & 0 & 0 & 0 & B\\
            \hline
            C & 0 & 0 & 0 & 0 & 1 & 1 & C\\
            D & 0 & 1 & 0 & 0 & 0 & 0 & D\\
            E & 0 & 0 & 0 & 1 & 0 & 0 & E\\
              & C & D & E & F & G & H & \\
        \end{pNiceArray}
    \end{gather*}
        \caption{The flow matrix after switching vertex $D$ from \ref{fig:ex233 changed} to be $Z$ labelled. Columns given by vertices $C$, $D$, $E$, $F$, $H$ form a basis, so output $G$ can be removed (changed to be $Z$ labelled).}
    \end{subfigure}
    
    \caption{Examples of labelled open graphs with more outputs than inputs, and how the number of outputs can be reduced to match the number of inputs.}
    \label{fig:ex233 outputs}
\end{figure}

\end{document}